\documentclass[twocolumn,showpacs,pra,aps,footnotes]{revtex4-1}
\usepackage{graphicx}
\usepackage{epstopdf}
\usepackage{amsmath}
\usepackage{amssymb}
\usepackage{epsfig}
\usepackage{amsthm}
\usepackage{color}
\newtheorem{thm}{Theorem}

\begin{document}
\title{Geometric Bell-like inequalities for steering}
\author{M. \. Zukowski}
\author{A.  Dutta}
\author{Z. Yin}
\affiliation{Institute of Theoretical Physics and Astrophysics, University of Gda\'{n}sk, 80-952 Gda\'{n}sk, Poland}

\begin{abstract}
Many of the  standard Bell inequalities (e.g., CHSH) are not effective for detection of quantum correlations which allow for steering, because for a wide range  of such correlations they are not violated. We present Bell-like inequalities which have lower bounds for non-steering correlations than for local causal models. The inequalities involve all possible measurement settings at each side. We arrive at interesting and elegant conditions for steerability of arbitrary two-qubit states.
\end{abstract}

\maketitle

\newcommand{\bra}[1]{\langle #1\vert}
\newcommand{\ket}[1]{\vert #1\rangle}
\newcommand{\abs}[1]{\vert#1\vert}
\newcommand{\avg}[1]{\langle#1\rangle}
\newcommand{\braket}[2]{\langle{#1}|{#2}\rangle}
\newcommand{\commute}[2]{\left[{#1},{#2}\right]}
\section{Introduction}

The concept of steering was introduced by Schr\"{o}dinger ~\cite{ Schrodinger} in 1935.  Recently the issue  was revisited in many papers. Wiseman {\it et al}~\cite{{WJD2007}} show that quantum steering is a notion of non-classicality  which is placed  between entanglement and violation of Bell inequality. There are only few steering inequalities known. For example,  linear steering inequality \cite{SJWP2010,SNC2013}, steering inequality due to an ``All-Versus-Nothing" argument \cite{CYWSCKO2013} ,  steering inequality from uncertainty principle \cite{SBWCH2013}, and  fine-grained uncertainty relation \cite{pramanik}.  In this work we address the problem of relation of Bell inequalities with steering. We show that some of the standard inequalities, like the CHSH ones \cite{CHSH}, detect steering only when it is associated with violation of their bound for local realistic (local causal) theories. Thus they cannot detect steering property of states which allow steering but do not violate local realism (see e.g. Ref.  \cite{{WJD2007}}; such a class in non-empty, and as a matter of fact it forms a subset of entangled states of a considerable ``volume"). However, we  find here a class of Bell inequalities which are capable to detect steering property for states which do {\em not} violate their local causal/realistic bound. This is done by showing that for states which belong to a non-steering class one can derive a new bound for the expression forming the inequality, which is lower than the bound for the local causal theories.
As a result, we  introduce a steering inequality which is inspired by  geometric Bell inequalities considered in \cite{MZ, KZ, badziag}.

\subsection{Steering}
We give a short introduction to  quantum  steering. The main idea of steering can  be formulated in the following way. Imagine  two communicating  partners Alice and Bob who supposedly share an entangled state. 
If this is so, and the entangled state belongs to the class allowing steering,  Alice, when asked by Bob to perform a measurement of his choice, and to report her outcome, can produce a family of states on Bob's side, in a way which cannot be reproduced by her sending  him  particles (before he asks her to do the measurements), in various states, which are {\em not} entangled with a quantum system in her lab.
Note, that not all entangled states allow steering, \cite{{WJD2007}}, and that every state which violates a Bell inequality allows steering (this will be evident after reading the next subsections).

 Imagine Alice and Bob share an ensemble of identical quantum states. The subsystems of Alice will be denoted by $1$, while Bob's by $2$.  A state $\varrho^{(12)}$ of a pair of subsystems does {\em not} allow steering by Alice of Bob's states if  one can write the states of Bob as
 \begin{equation}
\label{steering1}
\varrho^{(2)}_{a|x}=\sum_\lambda{p_\lambda}P(a|x,\lambda)\varrho^{(2)}_{\lambda},
\end{equation}
where  $x$ is the  measurement (supposedly  performed) by Alice for which Bob asked, whereas $a$ her declared result. The parameter $\lambda$  is just an index, of whatever form, even a continuous variable, or set of them, which parametrizes the states sent by Alice. The symbols  $P(a|x,\lambda)$ and $p_\lambda$ stand for probabilities (one has $\sum_\lambda p_\lambda=1$), whereas $\varrho^{(2)}_{\lambda}$ is some state of the particle received by Bob form Alice  (often called ``hidden state").
Note that  $P(a|x, \lambda)$ can be treated as some realization of a hidden variable model, namely a probability distribution of  measurement outcomes $a$ under settings $x$ for a hidden variable value given by  $\lambda,$ with the probability of the hidden variables given by  $p_\lambda$.
The structure of the state in Eq.  (\ref{steering1})  allows one  to consider all kinds of hidden variable models on  Alice's  side, but in Bob's side  the description is restricted to  hidden local quantum states. 
If such a combination of ``local hidden variable (LHV) and local hidden state"  (LHS) models $\eqref{steering1}$ can be constructed, then the state is not steerable.
 If a  model of type $\eqref{steering1}$ for $\varrho^{(2)}_{a|x}$ does not exist for Bob's states, then Alice can steer his states.


 Let's say Alice and Bob share a quantum state $\varrho^{(12)}$ which allows steering and  Alice's (generalized) measurement operators are $M^{(1)}_{a|x}$, where $a$ is an outcome for measurements defined by the setting  $x$. If outcome $a$ happens, the state of Bob's system spontaneously collapses to  $N \varrho^{(2)}_{a|x},$ where $N$ is a normalization factor, and
\begin{equation}
\label{steering}
\varrho^{(2)}_{a|x}=\text{tr}_{(1)}(\varrho^{(12)}M^{(1)}_{a|x}\otimes I^{(2)}).
\end{equation}
In other words,  measurements in Alice's side lead to a set of  states of Bob's systems  which can be modeled only by the quantum formula (\ref{steering}). 

To test if  Alice really can steer the state,  Bob may ask  her to make some (generalized) measurements of arbitrary kind of  his choice, and to report her outcomes. After many repetitions, when big enough ensembles are formed for given  measurements on Alice's side and specific results,  by a local state tomography Bob can check if the states are indeed  $\varrho^{(2)}_{a|x}$, and then check by some methods if they are obtainable only via steering by Alice. One of such methods will be suggested in this article.

\subsection{Bell inequalities vs steering}

Let us describe, from the point of view of mathematical formalism, the relation between  Bell inequalities and the problem of steering~\cite{WJD2007, B2014,  CYWSCKO2013, SBWCH2013, SJWP2010, SNC2013}.  This will be presented here for qubit systems, as our results will applicable to such systems, and for von Neumann type measurements. This allows us to parametrize the local settings by Bloch vectors.
\begin{itemize}
\item
Bell inequalities hold for  local  causal models, for which,  in the case of two particles (here for simplicity qubits), the correlations, in terms of joint probabilities of two pairs of local results $P(r_1, r_2|\vec{a}, \vec{b})$,  can be described by :
\begin{equation}
\label{bell}
P(r_1, r_2|\vec{a}, \vec{b})_{LHV}=\int p_\lambda P(r_1|\vec{a}, \lambda)P(r_2|\vec{b}, \lambda)d\lambda,
\end{equation}
where $r_1, r_2$ are local results, $\lambda$ is a cause (essentially, a  hidden variable) and $\vec{a}$ and $\vec{b}$ denote local settings of the measuring devices.  $P(r_1|\vec{a}, \lambda)$ and  $P(r_2|\vec{b}, \lambda)$ are some (classical) probabilities. 
\item 
The {\em steering} property can be defined as the {\em non-existence} of the following model of correlations
\begin{equation}
\label{steering}
P(r_1, r_2|\vec{a}, \vec{b})_{NS}=\int p_\lambda P(r_1|\vec{a}, \lambda)Tr[\hat{\pi}(r_2|\vec{b})\varrho^{(2)}_\lambda] d\lambda,
\end{equation}
where $\hat{\pi}(r_2|\vec{b})$ is the projection operator for an observable parametrized by the  setting $\vec{b}$, which is associated with the eigenvalue $r_2$(the result), and $\varrho^{(2)}_\lambda$ is some pure state of Bob's  system $2,$ parametrized by some set of parameters $\lambda.$  One can alternatively call correlations describable by (\ref{steering}) as allowing a hidden state description. Note that the model (\ref{steering}) is a special case of 
(\ref{bell}), in which the probability $P(r_2|\vec{b}, \lambda)$ is replaced by the quantum formula $Tr[\hat{\pi}(r_2|\vec{b})\varrho^{(2)}_\lambda] $.
\end{itemize}

In the  case of models (\ref{steering}) one does not put any restrictions on the form of $ P(r_1|\vec{a}, \lambda)$, other than the ones assumed for the case of (\ref{bell}), but on Bob's side projective measurements on quantum states are assumed. A two system state is non steerable when the local measurement statistics  can be described by a local hidden variable model on one side (Alice) and local hidden states on the other side (Bob)  $\eqref{steering}.$ 

{ { 
If we have a violation of Bell inequality, then definitely the quantum state giving the correlations allows steering. However,  in many cases we find that the given state does not violate well  known Bell inequalities and still allows steering. {This is because the non-steerable correlations can be thought of as described by a {\em restricted} local hidden model which on one side allows only probabilities which are allowed by quantum formalism applied to a single particle. Compare (\ref{bell}) and (\ref{steering}).  

 The famous CHSH inequality \cite{CHSH} fails to distinguish between models (\ref{bell}) and (\ref{steering}). If we try to re-derive  it for non-steering correlations, given by  $\eqref{steering}$, we get the following basic algebraic relation
\begin{equation}
\label{CHSH}
I_1(\vec{a}_1, \vec{\lambda})[\vec{b}_1\cdot \vec{\lambda}+\vec{b}_2\cdot \vec{\lambda}]+I_1(\vec{a}_2, \vec{\lambda})[\vec{b}_1\cdot \vec{\lambda}-\vec{b}_2\cdot \vec{\lambda}],
\end{equation}
where  from now on $\vec{\lambda}$ represents a Bloch vector of a single qubit state $\varrho^{(2)}_{\vec{\lambda}}$, $I_1(\vec{a}, \vec{\lambda})=P(+1|\vec{a}, \vec{\lambda})-P(-1|\vec{a}, \vec{\lambda}),$ etc., and finally 
\begin{equation}\label {QUANTUM}
\vec{b}\cdot\vec{\lambda}=\text{Tr}[\vec{b}\cdot\vec{\sigma}^{(2)}\varrho^{(2)}_{\vec{\lambda}}],
\end{equation}
 where in turn $\vec{\sigma}^{(2)}$ is Pauli vector, while $\vec{b}_i$ and $\vec{a}_j$ denote the  vectors defining the settings. Note that expressions (\ref{QUANTUM}) replace in (\ref{CHSH}) the usual $I_2(\vec{b}, \vec{\lambda})=P(+1|\vec{b}, \vec{\lambda})-P(-1|\vec{b}, \vec{\lambda})$ for  the case of local hidden variables. Unfortunately the maximum of this algebraic relation, $B_{NS}$, is 2, just as in the case of local hidden variables ($B_{LHV}=2$).  Thus,  if we replace in the algebraic expression of  CHSH the local causal models by  a model involving  quantum correlations without the steering property, what leads to  (\ref{CHSH}),  the upper bound of this expression $B_{NS}$, remains the same as for local causal theories, $B_{LHV}=B_{NS}$. Due to this fact,  CHSH inequalities cannot detect steering which does not violate local realism (all that for the specific CHSH scenario of two-settings per observer).

We shall show that,  if one uses non-standard Bell  inequalities  of Refs. ~\cite{MZ, KZ}, involving all possible settings on both sides of the testing experiment, their bound  $B_{NS}$ for {\em non-steering} quantum correlations (\ref{steering})  is lower than their bound $B_{LHV}$ for local causal correlations (\ref{bell}). 
This means the following: we shall have a functional  expression  $L[E]$ of values in real numbers, which depends linearly  on  correlation functions $E(\vec{a}, \vec{b})=\sum_{r_1,r_2=\pm1}r_1r_2P(r_1,r_2|\vec{a}, \vec{b})$. If the correlation functions $E_{LHV}$ allow  local hidden variable  models, (\ref{bell}), one can find that one always has
\begin{equation} 
L[E_{LHV}]\leq B_{LHV}. 
\end{equation}
The number $B_{LHV}$ is the bound of a Bell inequality based on the functional $L[E]$. The value of $B_{LHV}$, for the specific $L[E]$ to be  studied in the next Section, was derived in \cite{KZ}. 
One can also find in \cite{KZ} that this bound is violated by specific quantum predictions. Thus the inequalities are ``relevant".  We shall show here that, if one restricts the models to the ones of the non-steering class (\ref{steering}),  
one has  
\begin{equation} 
L[E_{NS}]\leq B_{NS}<B_{LHV}. 
\end{equation}
Thus, the new inequality, $L[E_{NS}]\leq B_{NS},$ can detect steering property in quantum correlations which {\em do not} violate  inequalities of the same kind  for the Bell problem. 

\subsection{Geometric approach leading to inequalities of Ref. \cite{KZ}}}

 The approach ~\cite{MZ, KZ, badziag} is  based on  the trivial fact that if one has two (real) vectors $v$ and $w,$ then,   if $(v, w)<||w||^2,$ one has  $v\neq w.$ In  simple words, if a scalar product of two vectors is less than the squared norm of one of them then they cannot be equal.

Following this idea for two qubit correlations,  we shall compare the quantum correlation function $E_Q(\vec{a}, \vec{b})$, which is always precisely defined for the given observables and the given quantum state (and thus is treated here a known function) with the generic correlation function for non-steerable states, of which it is only known that, by (\ref{steering}),  in general it must have the following structure
 $E_{NS}(\vec{a},\vec {b})=\int p_{\vec{\lambda}}I_A(\vec{a}, \vec{\lambda})\vec{b}\cdot\vec{\lambda}d\vec{\lambda}.$ Here $p_{\vec{\lambda}}$ severs the same role as $p_\lambda$ in Eq. (\ref{steering}). Note that if $\varrho_\lambda$ in (\ref{steering}) are replaced by qubit states, then they can indexed by their Bloch vectors, $\vec{\lambda}$.

The scalar product can be defined as
\begin{equation}
\label{starting}
(E_{NS},E_Q)=\int \int E_{NS}(\vec{a}, \vec{b})E_{Q}(\vec{a}, \vec{b})d\Omega(\vec{a})d\Omega(\vec{b})
\end{equation}
(the integrations are over the Bloch spheres). It can be shown that, when treating $E_Q$ as fixed, the scalar product has an upper bound, $B_{NS}$. That is for any  $E_{NS}(\vec{a}, \vec{b}).$  one has 
\begin{equation}\label{COND}
 L(E_{NS})=(E_{NS},E_Q)\leq B_{NS}.
\end{equation} 
Note that this has  a formal structure of a Bell (like) inequality with function $E_Q(\vec{a},\vec{b})$ giving a (fixed, known) continuous set of coefficients. We  show that, for a wide range of states which do not violate the CHSH inequality or the inequality based on the same geometric concepts but applied to local causal models \cite{KZ}, one has
\begin{equation}
\label{ineq}
\text{B}_{NS}<\int \int E^2_{Q}(\vec{a}, \vec{b})d\Omega(\vec{a})d\Omega(\vec{b})=||E_Q||^2=L(E_Q),
\end{equation}
which means that they do not  have a  hidden state description, i.e. they allow steering. From relation (\ref{COND}) can derive a concise sufficient condition, which allows us tell that a quantum state allow steering (see below).}}

\section{Derivation of the steering inequalities}
We shall first give definitions of notions which will be used in the derivation. Next we shall move to the construction of the Bell-like inequality which is a necessary condition for the non-steering case.

\subsection{Correlation function for two-qubit states}

Suppose Alice performs a projective measurement $\vec{ m} \cdot \vec{\sigma}^{(1)}$, where $\vec{\sigma}^{(1)}$ is her vector built out of Pauli matrices, and $\vec{m}$ represents the unit  Bloch vector defining her measurement direction. The possible results are of course $\pm1$. Similarly  Bob performs projective measurements of $\vec{ n} \cdot \vec{ \sigma}^{(2)}$. The quantum correlation function is given by
\begin{equation}
\label{eq:Qcorrelation}
E_{Q}(\vec{ m},\vec{n}) = Tr ( \vec{ m} \cdot \vec{\sigma}^{(1)} \otimes \vec{ n} \cdot \vec{ \sigma}^{(2)} \varrho^{(12)} ),
\end{equation}
where $\varrho^{(12)}$ represents the two-qubit state.

An arbitrary density operator for two qubits can be written as
\begin{equation}
\label{eq:2d_state}
\varrho^{(12)}=\frac{1}{4}\sum_{\mu_1, \mu_2=0}^{3}T_{\mu_1\mu_2}\sigma^{(1)}_{\mu_1} \otimes \sigma^{(2)}_{\mu_2},
\end{equation}
where $\sigma^{(k)}_{\mu_n}$ for $\mu_n=i=1,2,3$, are the Pauli matrices forming the aforementioned vectors for observers $k=1,2$,  and  $\sigma^{(k)}_0 = \openone$. The components $ T_{\mu_1\mu_2}$ are real and given by $T_{\mu_1\mu_2}=  \textrm{Tr} [\varrho^{(12)} (\sigma^{(1)}_{\mu_1} \otimes \sigma^{(2)}_{\mu_2})]$. By equation \eqref{eq:Qcorrelation}, the correlation function between measurement outcomes on Alice and Bob's side reads
\begin{equation}
\label{correlation}
E_Q(\vec{m}, \vec{n})=\sum_{i, j=1}^{3}T_{ij} m_{i} n_{j},
\end{equation}
where $m_i$, and $n_j$ are Cartesian components of the Bloch vectors defining measurement directions. The set $T_{ij}$ for $i,j=1,2,3$ forms what is often called the correlation tensor (or  matrix) of the state $\varrho^{(12)}.$

\subsection{Correlation function model for states which do not allow steering}

If there is a LHS model describing the correlation function (\ref{eq:Qcorrelation}), then it has the following structure
\begin{equation}\label{eq:Lcorrelation}
E_{NS}(\vec{m},\vec{n}) = \sum_{\lambda} p_\lambda I_1( \vec{m}, \lambda) Q(\vec {n}, \lambda),
\end{equation}
where $I_1(\vec{m}, \lambda) =  P(1|\vec{m},\lambda) -  P(-1|\vec {m},\lambda)$ and $Q(\vec {n}, \lambda) = Tr (\vec{ n} \cdot \vec{ \sigma}^{(2)} \varrho^{(2)}_{\lambda}).$  Obviously, as mixed states are probabilistic convex  combinations of pure states, without any loss of generality,   one can demand in  the construction of model  that all density operators $\varrho^{(2)}_{\lambda}$ represent pure states. This will be followed below.

\subsection{A geometric criterion for quantum steering}

 Our geometric criterion boils down to the following.
If the correlation function $E_{Q}$ for the given state $\varrho^{(12)}$, calculated according to quantum rules with formula (\ref{eq:Qcorrelation}), has the following property: for {\em any } $E_{NS}$ one has
\begin{equation}
\label{eq:geoineq}
|| E_{Q}||^2 >  (E_{Q}, E_{NS}),
\end{equation}
then $E_Q$ cannot be described by a local hidden state  model (\ref{steering}), i.e, $\varrho^{(12)}$ allows steering (from Alice to Bob). Of course to establish such a fact one must find  $B_{NS}$ which is given
by 
\begin{equation}
B_{NS}=Max_{E_{NS}} (E_{Q}, E_{NS}).
\end{equation}
In this way we get an effective (state dependent) steering inequality
\begin{equation}\label{NS-INEQ}
 (E_{Q}, E_{NS})\leq B_{NS},
\end{equation}
which when {\em violated} points that the given state allows steering. This means that, if one replaces in it $E_{NS}$ by $E_Q$, and the resulting value $(E_Q,E_Q)$ exceeds
the bound $B_{NS}$, the correlations $E_Q$ allow steering. Note further, that one can also  treat $E_Q$ in (\ref{NS-INEQ}) as just a continuous set of coefficients in a linear (functional) inequality, and test some different quantum correlation function $E_q$, defined for some quantum state $\varrho^{(12)}_q$, also by replacing by it in the inequality the non-steering functions $E_{NS}$. This would give $(E_Q,E_q)$, and if one has  $(E_Q,E_q)> B_{NS}$, that is the inequality is violated, then correlations described by $E_q$ also allow steering.

One can choose among infinitely many definitions of the scalar product. Here we   consider the simplest one.
Namely,  if we denote the Bloch spheres of local measurement setting of Alice (and Bob) by $\Omega$,  the natural  Hilbert space in which we consider  our correlation functions  is the real space $L^2(\Omega \times \Omega), $ with a scalar product for two functions $f(\vec{m}, \vec{n})$ and
$g(\vec{m}, \vec{n})$ given by
\begin{equation}
(f,g)=\int_{\Omega \times \Omega}  f(\vec{m}, \vec{n})g(\vec{m}, \vec{n})d\Omega(\vec{m}) d\Omega(\vec{n}),
\end{equation}
where this integration measures are the usual (rotationally invariant, ``Haar") measures on the spheres (in simple words, integrations over a solid angle).

 We have following theorem:

\begin{thm}\label{thm:geo}
 If a state endowed with correlation function $E_Q(\vec{m}, \vec{n})=\sum_{i, j=1}^{3}T_{ij} m_{i} n_{j}$ is non-steerable, then
\begin{equation}
Max_{\vec{m}, \vec{n}}E(\vec{m}, \vec{n})\geq\frac{2}{3}\sum_{i, j=1}^{3}T^2_{ij}.
\end{equation}
\end{thm}

\begin{proof}
Let us use spherical coordinates to express the vectors, e.g. for $\vec{n}$ we put  $(\sin{\theta_n}\cos{\phi_n}, \sin{\theta_n}\sin{\phi_n}, \cos{\theta_n}).$ Then according to equation \eqref{eq:Lcorrelation} and \eqref{correlation}, the right hand of \eqref{eq:geoineq} is
\begin{equation}
\label{maximum}
\sum_{\vec{\lambda}} p_{\vec{\lambda}} \int_{\Omega \times \Omega} I_A (\vec{m}, \vec{\lambda}) (\vec{m} \cdot T \vec{n}) ( \vec{n}\cdot \vec{\lambda}) d\Omega(\vec{m}) d\Omega(\vec{n}),
\end{equation}
where the vector $\vec{\lambda}$ is the Bloch vector corresponding to a pure state \cite{reason} and $d\Omega(\vec{n}) $ is a rotationally invariant measure on Bloch sphere of unit radius, namely $\sin\theta_n d\theta_n d\phi_n.$ The  functions $ n_k(\theta_n, \phi_n)$ which give the components $k=1,2,3$ of $\vec{n}$ in spherical coordinates  obey the following orthogonality relation
\begin{eqnarray}
\label{orthogonality}
\int_\Omega n_k(\theta_n, \phi_n)n_l(\theta_n, \phi_n)d\Omega(\vec{n})=\frac{4\pi}{3}\delta_{kl}.
\end{eqnarray}
Thus,
\begin{equation}\label{eq:intofvector}
\int_{\Omega} (\vec{m} \cdot T \vec{n}) (\vec{n} \cdot \vec{\lambda}) d\Omega(\vec{n })= \frac{4\pi}{3}( \vec{m} \cdot T \vec{\lambda}),
\end{equation}
and  \eqref{maximum} can be written as
\begin{equation}\label{eq:maximal}
\frac{4\pi}{3} \sum_{\vec{\lambda}} p_{\vec{\lambda}} \int_{\Omega} I_A (\vec{m},\lambda ) ( \vec{m} \cdot T \vec{\lambda}) d\Omega(\vec{m}).
\end{equation}

$T$ is a $3\times 3$ real matrix. Let $\Sigma$ be its  singular value decomposition (i.e. Schmidt decomposition). One has $T = U^T \Sigma V,$ where $U$ and $V$ are specific  $3\times 3$ orthogonal matrices representing rotations (proper and improper) of the local coordinates of Alice and Bob, respectively.  The matrix $\Sigma$ is  diagonal,  with diagonal entries which we denote as $(T_1,T_2,T_3),$ and customarily one assumes that  $T_1 \geq T_2 \geq T_3 \geq 0.$ Thus in matrix notation, with Bloch vectors treated as column matrices 
\begin{equation}
\vec{m} \cdot T \vec{\lambda} = (U \vec{m} )^T\Sigma V \vec{\lambda}.
\end{equation}
It is important to remember that $T_1=Max_{\vec{m}, \vec{n}}\vec{m} \cdot T \vec{n}$.

By changing integration variable to  $\vec{m}' = U\vec{m}$ and putting $\vec{\lambda}' = V \vec{\lambda},$ since $d\Omega(\vec{m}) $ is a rotationally invariant measure, we have:
\begin{equation}\label{eq:maxima2}
( E_{Q},E_{NS}) = \frac{4\pi}{3}  \sum_{\vec{\lambda}} p_{\vec{\lambda}} \int_{\Omega} I_A (\vec{m},\vec{\lambda}) ( \vec{m} \cdot \Sigma \vec{\lambda}) d\Omega(\vec{m}),
\end{equation}
where we have dropped the primes.

To estimate the above we can first  calculate the maximum length of the projection of  $I_A(\vec{m}, \vec{\lambda})$ on a three dimensional subspace of $L^2(\Omega),$ spanned by the components of $\vec{m}$, that is by functions $m_1=\sin\theta_m\cos\phi_m$, $m_2=\sin\theta_m\sin\phi_m$, and $m_3=\cos\theta_m$. We shall denote this subspace as $S^{(3)}$.

As any vector $\vec{v}$ is equal to its length (norm) $||\vec{v}||$ times the unit
directional vector associated with it, $\vec{v}/||\vec{v}||$, the projection of
$I_A(\vec{m}, \vec{\lambda})$, denoted here as $I_A^{||}(\vec{m}, \vec{\lambda})$ can be always put in the
following form
\begin{equation}
I_A^{||}(\vec{m}, \vec{\lambda})=||I_A^{||}(
\vec{\lambda})||{u}(\vec{m},\vec{ \lambda}),
\end{equation}
where
${u}(\vec{m}, \vec{\lambda})$ is a unit vector (normalized function) in $S^{(3)}$,  and  $||I_A^{||}(
\vec{\lambda})||$ is the norm of the projection of $I_A$ into  $S^{(3)}$. The maximal possible value of the norm  in turn can be  given by
\begin{equation}
M= Max_{I_A, u(\vec{\lambda})}( I_A, u (\vec{\lambda})), \label{PROJECTION}
\end{equation}
where the scalar product in $L^2(\Omega)$ is defined
$$
( I_A, u (\vec{\lambda}))= \int_{\Omega}I_A(\vec{m},\vec{ \lambda}){u}(\vec{m},\vec{ \lambda})d\Omega(\vec{m}).
$$
Any $u(\vec{\lambda})$ can be put as
$$\sqrt{\frac{3}{4\pi}}(\sin\alpha_{\vec{\lambda}}\cos\beta_{\vec{\lambda}} m_1
+\sin\alpha_{\vec{\lambda}}\sin\beta_{\vec{\lambda}}  m_2+ \cos\beta_{\vec{\lambda}} m_3)
$$
where $\alpha_{\vec{\lambda}}$ and $\beta_{\vec{\lambda}}$ are some angles defining a unit 3 dimensional  vector $\vec{w}_{\vec{\lambda}}$ in $R^3$, given by 
$$(\sin\alpha_{\vec{\lambda}}\cos\beta_{\vec{\lambda}},
\sin\alpha_{\vec{\lambda}}\sin\beta_{\vec{\lambda}},\cos\beta_{\vec{\lambda}}).
$$
 Therefore we have
 $u(\vec{m}, \vec{\lambda})=\sqrt{\frac{3}{4\pi}} \vec{w}_{\vec{\lambda}}\cdot \vec{m}.$ Thus 
the value  of (\ref{PROJECTION}) is given by maximum possible value of
\begin{equation}\label{INT}
\sqrt{\frac{3}{4\pi}} \int_\Omega I_A(\vec{m}, \vec{\lambda}) \vec{w}_{\vec{\lambda}}\cdot \vec{m}  d\Omega(\vec{m}).
\end{equation}
Note that   $-1\leq I_A(\vec{m}, \vec{\lambda})\leq1.$  When estimating the integral integral it optimal to   the $\hat{z}$ direction of the spherical coordinates as  $\hat{z}=\vec{w}_{\vec{\lambda}}$. With these two properties in mind one sees that the value of integral in expression (\ref{INT})  is maximally  $\int _{\Omega}|\cos\theta_m| d\Omega(\vec{m})=2\pi$.  Thus we  have $M=\sqrt{3\pi}=Max ||I^{||}_A(\vec{\lambda})||$.

  Now inserting $I_A^{||}(\vec{m}, \vec{\lambda})$,  instead of $I_A^{}(\vec{m}, \vec{\lambda})$  into \eqref{eq:maxima2} and next replacing it by  $M \sqrt{\frac{3}{4\pi}}\vec{w}_{\vec{\lambda}}\cdot \vec{m} $
one can estimate the integral in \eqref{eq:maxima2} as bounded by

\begin{equation}
\sqrt{3\pi}\sqrt{\frac{4\pi}{3}}Max_{\vec{\lambda}, \vec{w}_{\vec{\lambda}}}
\int_{\Omega} (\vec{w}_\lambda\cdot \vec{m} ) (\vec{m} \cdot \Sigma \vec{\lambda}) d\Omega(\vec{m}),
\end{equation}
which, using the same integration methods as before to reach (\ref{eq:intofvector}) can be simplified to
\begin{equation}
{2\pi}\left(\frac{4\pi}{3}\right)Max_{\vec{\lambda}, \vec{w}_{\vec{\lambda}}}
 \vec{w}_{\vec{\lambda}} \cdot \Sigma \vec{\lambda}= \frac{8\pi^2}{3} T_1.
\end{equation}
This ends our estimate of the maximal value of  $( E_Q, E_{NS})$.

We have a {\em Bell-like inequality}
\begin{equation}\label{FINAL}
(E_Q, E_{NS})\leq  \frac{8\pi^2}{3} T_1.
\end{equation}
It is linear with respect to $E_{NS}$, and is defined by a continuous set of coefficients given by the values of the correlation function $E_{Q}$. Note that it has an interesting feature, that if we replace $E_{NS}$ by $E_Q$ it becomes non-linear.

The final step is to calculate $||E_{Q}||^2=( E_Q, E_{Q}).$
Using orthogonality relation $\eqref{orthogonality}$ we obtain
\begin{equation}
\label{tensorsq}
(E_{Q},E_{Q}) = \frac{16\pi^2}{9}\sum_{j, k=1}^3 T^2_{jk}.
\end{equation}
Hence, the state is steerable if $T_1 <\frac{2}{3}\|T\|^2$, where $\|T\|^2=\sum_{j, k=1}^3 T^2_{jk}.$
\end{proof}

Note that similar condition for non existence of  local causal models, which was shown in \cite{KZ}, reads $T_1<\left (\frac{2}{3}\right )^2||T||^2,$ because $|( E_Q,E_{LHV})|\leq(2\pi)^2T_1,$ where $E_{LHV}$ is a correlation function modelable by local hidden variables. Note that local hidden states give a bound to the scalar product which is by a factor of $\frac{2}{3}$ lower than in the case of local causality. Note further, that a condition, of a  similar type,  for entanglement, which employs the same scalar product, reads $T_1 <\|T\|^2,$ see \cite{badziag}. Thus each step towards the  condition for entanglement changes the coefficient in front of $\|T\|^2$ by $\frac{3}{2}$.

\section{Examples}

\subsection{Werner states}
The  mixtures of the singlet $\left|\psi^{-}\right\rangle= \frac{1}{\sqrt{2}}(|01\rangle-|10\rangle)$ with white noise of the form:
\begin{equation}\label{werner}
\varrho_v^w= v |\psi^-\rangle\langle \psi^-| + (1-v)\frac{\openone}{4}.
\end{equation}
are known to be entangled if and only if $v>\frac{1}{3}$, and they violate the CHSH inequality when $v> \frac{1}{\sqrt{2}}.$ The condition for their steerability  has been shown in Ref. \cite{WJD2007}  and reads $v>\frac{1}{2}.$  Here we can use Theorem \ref{thm:geo} to  easily detect the steerability of $\varrho_v^w.$ Simply the correlation tensor $T$ of $\varrho_v^w$, in its Schmidt form\cite{reason2},
  is a diagonal matrix with entries $-v(1,1,1),$ and $T_1 = v$. So a simple algebra  recovers the threshold of $v=\frac{1}{2}$.

\subsection{Pure non-maximally entangled state with noise}
Consider a general pure state $|\psi_\alpha\rangle $ in Schmidt decomposition $\cos\frac{\alpha}{2}|00\rangle+\sin\frac{\alpha}{2}|11\rangle$, with $0\leq \alpha \leq \pi$. To get a direct comparison with the singlet studied above, let us transform it into:  $\sin\frac{\alpha}{2}|01\rangle-\cos\frac{\alpha}{2}|10\rangle$. The correlation tensor is then diagonal with
$T_{33}=-1$, and $T_{22}=T_{11}=-\sin{\alpha}$. If one considers "white noise" admixtures like (\ref{werner}),  with $|\psi^-\rangle$ replaced by  $|\psi_\alpha\rangle$, then $T_1=v$ and $||T||^2=v^2(1+2\sin^2\alpha)$. Therefore the sufficient condition for steering reads
\begin{equation}
v\geq \frac{3}{2(1+2\sin^2\alpha)}.
\end{equation}
Thus the noisy generalized Werner state allows for steering for $\alpha >\frac{\pi}{6}$.

\section{Final remarks}
Note that by using a weight function in the definition of the scalar product and thus also of the norm, one can derive an infinite number of conditions of a similar kind. To be more specific, one can always use a weight $w(\vec{m},\vec{n})$, to redefine the scalar product 
\begin{equation}
( f,g)_w=\int_{\Omega \times \Omega}  f(\vec{m}, \vec{n})g(\vec{m}, \vec{n})w(\vec{m},\vec{n})d\Omega(\vec{m}) d\Omega(\vec{n}).
\end{equation}
The weight must give for any $f$ 
\begin{equation}
(f,f)_w=\int_{\Omega \times \Omega}  f(\vec{m}, \vec{n})f(\vec{m}, \vec{n})w(\vec{m},\vec{n})d\Omega(\vec{m}) d\Omega(\vec{n})\geq 0.
\end{equation}
In this way one can get an infinite family of conditions, much like the conditions for entanglement of Ref. \cite{badziag}. By finding an optimal $w$ for a {\em given} state one can get steering conditions which can be  better than the ones presented here. A further generalization of such an approach may additionally involve single particle measurements, e.g. the averages
\begin{equation}\label{SINGLE}
E_{Q}(\vec{ m}) = Tr \big( \vec{ m} \cdot \vec{\sigma} \otimes  \openone \cdot \vec{ \sigma} \varrho^{(12)}\big).
\end{equation}
This will be studied in a forthcoming publication.

Finally, we would like to mention that conditions, of a different kind,  for steering, which employ the correlation tensor,  were independently derived in Ref. \cite{WISEMAN}.

\section{Acknowledgments}
The work is a part of  Polish Ministry of Science and
Higher Education Grant no. IdP2011 000361.  MZ and ZY  acknowledge support  by TEAM project of FNP and ERC AdG grant QOLAPS. A.D. is supported within the International Ph.D. Project “Physics
of future quantum-based information technologies,” Grant
MPD/2009-3/4 of the Foundation for Polish Science (FNP).


\begin{thebibliography}{99}
\bibitem{reason}
Since our purpose is to maximalize the right hand of \eqref{eq:geoineq}, so we can always assume the hidden state $\varrho_{\lambda}$ is pure.
\bibitem{reason2}If one requires that the bi-orthogonal Schmidt expansion is to be expressed in the form of proper right-handed set of three Cartesian unit vectors, the entries on the diagonal on the Schmidt decomposed tensor are for the singlet state negative. We use this convention here. Still the definition of $T_1$ which bounds the left hand side of the inequality (\ref{FINAL}) does not require this convention, thus for the singlet $T_1=1$,  and for the Werner state it reads  $v$.
\end{thebibliography}
\end{document}